\documentclass{llncs}[12pt]

\usepackage{amsmath}
\usepackage{graphicx}
\usepackage{amssymb}
\usepackage{float}

\floatstyle{ruled}
\newfloat{algorithm}{htbp}{loa}
\floatname{algorithm}{Algorithm}

\newcommand\F{\ensuremath{\mathcal{F}}}
\newcommand\B{\ensuremath{\mathcal{B}}}
\newcommand\D{\ensuremath{\mathcal{D}}}
\newcommand\ORD{\hbox{\sc Ord}}
\newcommand\BORD{\hbox{\sc BiOrd}}
\newcommand\DORD{\hbox{\sc DiOrd}}

\begin{document}

\title{Ordering without forbidden patterns}

\author{Pavol Hell\inst{1}
\and Bojan Mohar\inst{1} \and Arash Rafiey\inst{1,2} }

\institute{Simon Fraser University, Burnaby, Canada,
\email{pavol,mohar,arashr@sfu.ca}\thanks{supported by NSERC
Canada}  \and  
Indiana State University, Indiana, USA}

%
%
\date{}
\maketitle

\begin{abstract}
Let $\F$ be a set of ordered patterns, i.e., graphs whose vertices
are linearly ordered. An $\F$-free ordering of the vertices of a
graph $H$ is a linear ordering of $V(H)$ such that none of the
patterns in $\F$ occurs as an induced ordered subgraph. We denote
by $\ORD(\F)$ the decision problem asking whether an input graph
admits an $\F$-free ordering; we also use $\ORD(\F)$ to denote the
class of graphs that do admit an $\F$-free ordering. It was
observed by Damaschke (and others) that many natural graph classes
can be described as $\ORD(\F)$ for sets $\F$ of small patterns
(with three or four vertices). This includes bipartite graphs,
split graphs, interval graphs, proper interval graphs, cographs,
comparability graphs, chordal graphs, strongly chordal graphs, and
so on. Damaschke also noted that for many sets $\F$ consisting of
patterns with three vertices, $\ORD(\F)$ is polynomial-time
solvable by known algorithms or their simple modifications. We
complete the picture by proving that {\em all} these problems can
be solved in polynomial time. In fact, we provide a single master
algorithm, i.e., we solve in polynomial time the problem $\ORD_3$
in which the input is a set $\F$ of patterns with at most three
vertices and a graph $H$, and the problem is to decide whether or
not $H$ admits an $\F$-free ordering of the vertices. Our
algorithm certifies non-membership by a forbidden substructure,
and thus provides a single forbidden structure characterization
for all the graph classes described by some $\ORD(\F)$ with $\F$
consisting of patterns with at most three vertices. This includes
bipartite graphs, split graphs, interval graphs, proper interval
graphs, chordal graphs, and comparability graphs. Many of the
problems $\ORD(\F)$ with $\F$ consisting of larger patterns have
been shown to be NP-complete by Duffus, Ginn, and R\" odl, and we add
two simple examples.

We also discuss a bipartite version of the problem, $\BORD(\F)$,
in which the input is a bipartite graph $H$ with a fixed
bipartition of the vertices, and we are given a set $\F$ of
bipartite patterns. We give a unified polynomial-time algorithm
for all problems $\BORD(\F)$ where $\F$ has at most four vertices,
i.e., we solve the analogous problem $\BORD_4$. This is also a
certifying algorithm, and it yields a unified forbidden
substructure characterization for all bipartite graph classes
described by some $\BORD(\F)$ with $\F$ consisting of bipartite
patterns with at most four vertices. This includes chordal
bipartite graphs, co-circular-arc bipartite graphs, and bipartite
permutation graphs. We also describe some examples of digraph
ordering problems and algorithms.

We conjecture that for every set $\F$ of forbidden patterns,
$\ORD(\F)$ is either polynomial or NP-complete.
\end{abstract}

\newpage

\section{Problem definition and motivation}

For every positive integer $k$ we write $[k]=\{1,2,\dots,k\}$,
$E_k =\{\{i,j\} \mid i,j \in [k], i \ne j\}$, and $\F_k =
2^{E_k}$. Each element in $\F_k$ can be viewed as a labelled graph
on vertex set $[k]$ and is called a {\em pattern} of {\em order}
$k$, or simply a {\em $k$-pattern}. Given an input graph $H$ and a
linear ordering $<$ of its vertices, we say that a pattern $F \in
\F_k$ {\em occurs} in $H$ (under the ordering $<$) if $H$ contains
vertices $v_1 < v_2 < \cdots < v_k$ such that the induced ordered
subgraph on these vertices is isomorphic to $F$, i.e., for every
$i,j\in [k]$, $v_iv_j\in E(H)$ if and only if $\{i,j\}\in F$. For
convenience, we shall henceforth write $ij$ to simplify notation
for unordered pairs $\{i,j\}$.

For a set $\F \subseteq \F_k$ we say that a linear ordering $<$ of
$V(H)$ is {\em $\F$-free} if none of the patterns in $\F$ occurs
in $<$. The problem $\ORD(\F)$ asks whether or not the input graph
$H$ has an $\F$-free ordering. We also consider the problem
$\ORD_k$ that asks, for an input $\F \subseteq \F_k$ and a graph
$H$, whether or not $H$ has an $\F$-free ordering.

The problems $\ORD(\F)$ can be viewed as 2-satisfiability problems
with additional ordering constraints, or as special ternary constraint
satisfaction problems. In neither case there are general algorithms
known for such problems, cf. also \cite{gutt}.

The problems $\ORD(\F)$ have been studied by Damaschke
\cite{dama}, Duffus, Ginn, and R\" odl \cite{dgr}, and others. In
particular, Damaschke lists many graph classes that can be
equivalently described as $\ORD(\F)$. For example
\cite{spinrad-lee}, it is well known that a graph $H$ is chordal
if and only if it admits an $\F$-free ordering for $\F$ consisting
of the single pattern $\{12,13\}$, and $H$ is an interval graph if
and only if it admits an $\F$-free ordering for $\F$ consisting of
the pattern $\{ \{13\}, \{13,23\} \}$.
Similar sets of patterns from $\F_3$ describe proper interval
graphs, bipartite graphs, split graphs, and comparability graphs
\cite{dama}. With patterns from $\F_4$ we can additionally
describe strongly chordal graphs \cite{martin}, circular-arc
graphs \cite{jayme}, and many other graph classes.

Analogous definitions apply to bipartite graphs: a {\em bipartite
pattern of order $k$} is a bipartite graph whose vertices in each
part of the bipartition are labelled by elements of $[\ell]$
respectively $[\ell']$, with $\ell+\ell' \le k$. We again denote
by $\B_k$ the set of all bipartite patterns of order $k$. The
problem $\BORD(\F)$ for a fixed $\F \subseteq \B_k$ asks whether
or not an input bipartite graph $H$ with a given bipartition
$V(H)=U \cup V$ admits an ordering of $U$ and of $V$ so that no
pattern from $\F$ occurs. We also define the corresponding problem
$\BORD_k$ in which both $\F \subseteq B_k$ and $H$ with $V(H)=U \cup V$
are part of the input.

Several known bipartite graph classes can be characterized as
$\BORD(\F)$ for $\F \subseteq \B_4$. For instance, for $\F =
\{11',31' \}$ (here $\ell = 3, \ell' = 1$), the class $\BORD(\F)$ consists
precisely of convex bipartite graphs, and $\F = \{ \{11',12', 21'\}, \{12',21'\},$
$\{12',21',22'\} \}$ (here $\ell = \ell' = 2$) similarly defines bipartite permutation
graphs (a.k.a., proper interval bigraphs) \cite{gregory,spinrad,spinrad1}.
One can similarly obtain the classes of chordal bipartite graphs,
and bipartite co-circular arc bigraphs \cite{hell-huang}.

\subsection*{Summary of our main results}

We show that $\ORD_3$ and $\BORD_4$ are solvable in polynomial
time. In particular, this completes the picture analyzed by
Damaschke \cite{dama}, and proves that all $\ORD(\F)$ with $\F \subseteq
\F_3$ are polynomial-time solvable; similarly, all $\BORD(\F)$
with $\F \subseteq \B_4$ are polynomial-time solvable.

We also discuss digraphs with forbidden patterns on three
vertices, and present two classes of digraphs for which our
algorithm can be deployed to obtain the desired ordering without
forbidden patterns.

We further describe sets $\F \subseteq \F_4$ for which $\ORD(\F)$ is
polynomial time solvable and other sets $\F \subseteq \F_4$ for which
the same problem is NP-complete. Many more NP-complete cases of
$\ORD(\F)$ are presented in \cite{dgr}; in particular, the authors
of \cite{dgr} conjecture that any $\F$ consisting of a single
2-connected pattern (other than a complete graph) yields an
NP-complete $\ORD(\F)$.

Our master algorithm for $\ORD_3$ provides a unified approach
to all recognition problems for classes $\ORD(\F)$ with $\F \subseteq
\F_3$, including all the well known graph classes mentioned earlier.
Our algorithm is a certifying algorithm, and so it also provides a
unified obstruction characterization for all these graph classes.
(We note that these graphs have different ad-hoc obstruction
characterizations \cite{golumbic}.) A similar situation occurs
with $\BORD(\F)$ with $\F \subseteq \B_4$ and classes characterized as
$\BORD(\F)$ with $\F \subseteq \B_4$, including the well known
classes of bipartite graphs mentioned earlier. We note that these
special graph classes received much attention in the past; efficient
recognition algorithms and structural characterizations can be found
in \cite{booth,tizina,corneil,fulkerson,habib,tarjan,spinrad,trotter} and
elsewhere, cf. \cite{spinrad-lee,golumbic}.

The algorithms use a novel idea of an auxiliary digraph. We
believe this will be useful in other situations, and we have used
similar digraphs in \cite{adjust-interval,arash}. The algorithm
for $\ORD_3$ runs in time $O(n^3)$ where $n$ is the number of
vertices of $H$ and in several cases (when the family $\F$ is
particularly nice) it runs in time $O(nm)$, where $m$ is the
number of edges of $H$. The algorithm for $\BORD_4$ runs in time
$O(n^4)$ and in several cases in time $O(n^2 m)$.  We note that
many of the special cases have recognition algorithms that are
$O(m+n)$, so we are definitely paying a price for having a unified
algorithm; we note that the auxiliary digraph we use has
$\Omega(nm)$ edges, so this technique is not likely to produce a
linear time unified algorithm.

We conjecture that for every set $\F$ of forbidden patterns,
$\ORD(\F)$ is either polynomial or NP-complete and provide some
additional evidence for this dichotomy.

A preliminary version of this paper, without most of the proofs
has appeared in ESA 2014 \cite{esa}.

\section{Algorithm for \ORD$_3$ on undirected graphs }\label{3pattern}

Consider an input graph $H$ and a set of patterns $\F \subseteq
\F_3$. Note that $\F$ imposes a constraint on any three vertices
$x,y,z$ of $H$. This means that whenever $(x,y,z)$ induce a
subgraph isomorphic to a pattern from $\F$ and $x$ is before $y$
then $z$ must not be after $y$.

We first construct an auxiliary digraph $H^+$, which we call a
{\em constraint digraph\/}. The vertex set of $H^+$ consists of
the ordered pairs $(x,y)\in V(H)\times V(H)$, $x \ne y$, and the
arcs of $H^+$ are defined as follows. There is an arc from $(x,y)$
to $(z,y)$ and an arc from $(y,z)$ to $(y,x)$ whenever the
vertices $x,y,z$ ordered as $x < y < z$ induce a forbidden pattern
in $\F$. We say that a pair $(x,y)$ {\em dominates} $(x',y')$ and
we write $(x,y) \rightarrow (x',y')$ if there is an arc from
$(x,y)$ to $(x',y')$ in $H^+$.

Consider a strong component $S$ of $H^+$. The dual component
$\overline{S}$ of $S$ consists of all the pairs $(y,x)$ where
$(x,y) \in S$. Note that if $(x,y) \rightarrow (x,z)$, then $(z,x)
\rightarrow (y,x)$.

There are two operations that appear naturally when dealing with
orderings and forbidden patterns \cite{dama}. If we replace each
pattern in $\F$ with its {\em complement} (change edges to
nonedges and vice versa), thus obtaining a set $\overline{\F}$,
then a linear ordering of $V(H)$ is $\F$-free for $H$ if and only
if it is $\overline{\F}$-free for the complementary graph
$\overline H$. Another equivalence is obtained by replacing $\F$
with patterns that represent the same induced subgraphs but with
the reversed order, e.g., replacing $\{12,13\}$ by $\{32,31\}$.
Then a linear ordering will be $\F$-free if and only if the
reverse ordering will be free of the reversed patterns. We will
rely on these two properties in some of our proofs.

In general, the structure of the digraph $H^+$ depends on the
patterns. It is easy to see that if $\{12,23\}$ or $\{13\}$ is the
only forbidden pattern in $\F\subset \F_3$, then $(u,v)(u',v')$ is
an arc of $H^+$ if and only if $(u',v')(u,v)$ is an arc of $H^+$,
i.e. $(u,v)(u',v')$ is a symmetric arc of $H^+$ and hence $H^+$ is
a graph. On the other hand, if $\{12,13\}$ is the only forbidden
pattern in $\F$, then $H^+$ is a digraph without digons and if
$(u,v)(u',v')$ is an arc, then $(u',v')(u,v)$ is not an arc of
$H^+$.

If all pairs $(x_0,x_1),(x_1,x_2),...,(x_{n-1},x_n),\allowbreak
(x_n,x_0)$, $n\ge1$, are in the same subset $D$ of $V(H^+)$ then
we say that $(x_0,x_1),(x_1,x_2),...,\allowbreak
(x_{n-1},x_n),\allowbreak (x_n,x_0)$ is a {\em circuit\/} in $D$.

\begin{lemma}\label{if-circuit}
Let $\F \subseteq \F_3$ and let $H^+$ be the constraint digraph of
$H$ with respect to $\F$. If there exists a circuit in a strong
component $S$ of $H^+$, then $H$ has no $\F$-free ordering.
\end{lemma}

\begin{proof} For a contradiction suppose $ <$ is an $\F$-free ordering.
Consider a circuit $(x_0,x_1),\dots,\allowbreak
(x_{n-1},x_n),\allowbreak (x_n,x_0)$ in $S$. Since $S$ is strong,
there is a directed path $P_i$ from $(x_i,x_{i+1})$ to
$(x_{i+1},x_{i+2})$ in $S$. If $x_i < x_{i+1}$ then following the
path $P_i$ in $S$ we conclude that we must have $x_{i+1} <
x_{i+2}$, and eventually by following each $P_j$, $0 \le j \le n$
we conclude that $x_i < x_{i+1} < \cdots < x_{i-1} < x_i$. This is
a contradiction. Thus we must have $x_{i+1} < x_i$.  Now there is
a path $P_i'$ in $\overline{S}$ and hence by following the
path $P_i'$ we must have $x_i < x_{i-1}$ and eventually conclude
that $x_{i+1} < x_{i} < \cdots < x_{i+2} < x_{i+1}$, yielding a
contradiction. \qed \end{proof}

\begin{lemma}\label{lem:empty pattern}
{\rm (a)} Suppose $\emptyset \in \F \subseteq \F_3$.

If $H$ contains an independent set of three vertices, then $H^+$
has a strong component with a circuit and $H$ has no $\F$-free
ordering.

Otherwise $H^+$ is the same for $\F$ and for
$(\F\setminus\{\emptyset\})$, and $H$ has an $\F$-free ordering if
and only if it has an $(\F\setminus\{\emptyset\})$-free ordering.

{\rm (b)} Suppose $\{12,13,23\} \in \F \subseteq \F_3$.

If $H$ contains a triangle, then $H^+$ has a strong component with
a circuit and $H$ has no $\F$-free ordering.

Otherwise $H^+$ is the same for $\F$ and for
$(\F\setminus\{\emptyset\})$, and $H$ has an $\F$-free ordering if
and only if it has an $(\F\setminus\{\{12,13,23\}\})$-free
ordering.
\end{lemma}

\begin{proof} We only prove part (a) since the proof of (b) is similar.

Let $a,b,c$ be pairwise nonadjacent vertices of $H$. If $\emptyset
\in \F$, then $(a,b)\rightarrow (c,b)$, and $(c,b)\rightarrow
(a,b)$, thus $(a,b)$ and $(c,b)$ are in the same strong component
of $H^+$. Similarly, we have $(a,b)\rightarrow (a,c)$, and
$(a,c)\rightarrow (a,b)$, thus $(a,b)$ and $(a,c)$ are in the same
strong component of $H^+$. By symmetry, applied to other pairs, we
conclude that all ordered pairs of two distinct vertices from the
set $\{a,b,c\}$ are in the same strong component $S$ of $H^+$.
Clearly, $(a,b),(b,a)$ is a circuit in $S$.

As for the second part of the claim, if $H$ has no independent set
of three vertices, then $\emptyset$ contributes no restriction to
orderings of $V(H)$, so both the claims follow. \qed \end{proof}

Our main result is the following theorem which implies that
$\ORD_3$ is solvable in polynomial time. (In fact, its proof will
amount to a polynomial-time algorithm to actually construct an
$\F$-free ordering if one exists.)

\begin{theorem}\label{main}
Let $\F \subseteq \F_3$ and let $H^+$ be the constraint digraph of
$H$ with respect to $\F$. Then $H$ has an $\F$-free ordering if
and only if no strong component of $H^+$ contains a circuit.
\end{theorem}

Theorem \ref{main} will follow from the correctness of our
algorithm for $\ORD_3$. The proof of correctness will be given
in a later section.

Note that Theorem \ref{main}
provides a universal forbidden substructure (namely a circuit in a
strong component of $H^+$) characterizing the membership in graph
classes as varied as chordal graphs, interval graphs, proper
interval graphs, comparability graphs, and co-comparability
graphs.

We say a strong component $S$ of $H^+$ is a {\em sink component\/}
if there is no arc from $S$ to a vertex outside $S$ in $H^+$.
Consider a subset $D$ of the pairs in $V(H^+)$. We say that a
strong component $S$ of $H^+\setminus (D \cup \overline{D})$ is
{\em green with respect to $D$} if there is no arc from an element
of $S$ to a vertex in $H^+ \setminus (D \cup \overline{D} \cup
S)$. This is equivalent to the condition that $S$ is a sink
component in $H^+\setminus (D\cup \overline D)$.

In the algorithm below, we start with an empty set $D$ and we
construct the final set $D$ step by step. After each step of the
algorithm, $D$ (and hence also $\overline D$) is the union of
vertex-sets of strong components of $H^+$ and neither $D$ nor
$\overline{D}$ contains a circuit. Each strong component $S$ of
$H^+$  either belongs to $D$ or $\overline{D}$ or $V(H^+)
\setminus (D \cup \overline{D})$. At the end of the algorithm $D
\cup \overline{D}$ is a partition of the vertices (pairs) in
$V(H^+)$ such that whenever $(x,y),(y,z) \in D$ then $(x,z) \in
D$. We will say that $D$ satisfies {\em transitivity condition}.
At the end of the algorithm we place $x$ before $y$ whenever
$(x,y) \in D$ and we obtain the desired ordering. We say a strong
component is {\em trivial} if it has only one element otherwise it
is called {\em non-trivial}.

\bigskip
\noindent
{\sc Ordering with forbidden $3$-patterns, $\ORD_3$}\\
\noindent
{\sc Input:} A graph $H$ and a set $\F \subseteq \F_3$ of forbidden patterns on three vertices\\
\noindent {\sc Output:} An $\F$-free ordering of the vertices of
$H$ or report that there is no such ordering.

\bigskip

\noindent {\sc Algorithm for $\ORD_3$}
\begin{enumerate}
\item If a strong component $S$ of $H^+$ contains a circuit then
report that no solution exists and exit.
      Otherwise, remove $\emptyset$ and $\{12,13,23\}$ from $\F$.
If $\F$ is empty after this step, then return any ordering of vertices of $H$ and stop.
\item Set $D$ to be the empty set. \item Choose a strong component
$S$ of $H^+$ that is green with respect to $D$. The choice is made
according to the following rules.

a) If $\F$ contains one of the forbidden patterns
$\{13,23\},\{12,13\},\{12,23\}$, then the priority is given to
strong components containing a pair $(x,y)$ with $xy \in E(H)$. If
there is a choice then it is preferred $S$ to be a trivial
component. Subject to these preferences, if there are several
candidates, then priority is given to the ones that are sink
components in $H^+$.

b) If $\F$ contains one of $\{12\},\{23\},\{13\}$, then priority
is given to a strong component $S$ containing $(x,y)$ with $xy
\not\in E(H)$. If there is a choice, then the priority is given to
trivial components, and if there are several candidates for $S$,
then preference is given to the sink components in $H^+$.

\item If by adding $S$ into $D$ we do not close a circuit, then we
add $S$ into $D$ and discard $\overline{S}$. Otherwise we add
$\overline{S}$ and its outsection (all vertices in $H^+$ that are
reachable from $\overline{S}$) into $D$ and discard $S$ and its
insection (the vertices that can reach $S$). Return to Step 3 if
there are some strong components of $H^+$ left.

\item For every $(x,y) \in D$, place $x$ before $y$ in the final
ordering.
\end{enumerate}

Our proof of the correctness of the algorithm will, in particular,
also imply Theorem \ref{main}.

\begin{corollary}
Each problem $\ORD(\F)$ with $\F \subseteq \F_3$ can be solved in
polynomial time.
\end{corollary}

{\bf Remark.} Our algorithm is linear in the size of $H^+$. The
number of edges in $H^+$ is at most $n^3$ since each pair $(x,y)$
has at most $n$ out-neighbors. Thus the algorithms runs in
$O(n^3)$, where $n = |V(H)|$. In some cases, e.g., when $|\F| =
1$, this can be improved to $O(nm)$, where $m = |E(H)|$.

\subsection{Obstruction Characterizations}

Many of the known graph classes discussed here have obstruction
characterizations, usually in terms of forbidden induced subgraphs
or some other forbidden substructures. A typical example is
chordal graphs, whose very definition is a forbidden induced
subgraph description: no induced cycles of length greater than
three. Interval graphs have been characterized by Lekkerkerker and
Boland \cite{lekker} as not having an induced cycle of length
greater than three, and no substructure called an asteroidal
triple. Proper interval graphs have been characterized by the
absence of induced cycles of length greater than three, and three
special graphs usually called net, tent, and claw \cite{wegner}.
Comparability graphs have a similar forbidden substructure
characterization \cite{gallai}.

The constraint digraph offers a natural way to define a common
obstruction characterization for all these graph classes. In fact,
Theorem \ref{main} can be viewed as an obstruction
characterization of $\ORD(\F)$ for any $\F \subseteq \F_3$, i.e., each
of these classes is characterized by the absence of a circuit in a
strong component of the constraint digraph. Moreover, our
algorithm is a certifying algorithm, in the sense that when it
fails, it identifies a circuit in a strong component of $H^+$.

For some of the sets $\F \subseteq \F_3$, we have an even simpler
forbidden substructure characterization. We say $x,y$ is an {\em
invertible pair} of $H$ if $(x,y)$ and $(y,x)$ belong to the same
strong component of $H^+$. (Thus an invertible pair is precisely a
circuit of length two.) We say $\F$ is {\em nice} if it is one of
the following sets
$$\{\{13\}\},\{\{12,23\}\}, \{\{13\}, \{13,23\}\}, \{\{13\},\{12,13\},\{13,23\}\}.$$

By following the correctness proof of our algorithm, it will be seen
that if $\F$ is nice, then the algorithm does not
create a circuit as long as every strong component $S$ of $H^+$
has $S \cap \overline{S} = \emptyset$. Thus we obtain the
following theorem for nice sets $\F$.

\begin{theorem}\label{intervaltm}
Suppose $\F$ is nice. A graph $H$ admits an $\F$-free ordering if
an only if it does not have an invertible pair. \qed
\end{theorem}

In fact the correctness proof will show that if there is any circuit
in a strong component of $H^+$, then there is also a circuit of
length two.

Theorem \ref{intervaltm} applies to, amongst others, interval
graphs, proper interval graphs, comparability graphs and
co-comparability graphs.

\section{Bipartite graphs}\label{bipartite}

In this section we consider bipartite graphs $H$ with a fixed
bipartition $U\cup V$. We prove that $\BORD_4$ is polynomial-time
solvable, and so $\BORD(\F)$ is polynomial-time solvable for each
$\F \subseteq \B_4$. Each forbidden pattern $F \in \F $ imposes
constraints for those 4-tuples of vertices that induce a subgraph
isomorphic to $\F$. We construct an auxiliary digraph $H^+$, that
we also call a {\em constraint digraph\/}. The vertex set of $H^+$
consists of the pairs $(x,y) \in (U \times U)\cup (V\times V)$,
where $x \ne y$, and the arc-set of $H^+$ is defined as follows.

There is an arc from $(x,y)$ to $(z,y)$ and an arc from $(y,z)$ to
$(y,x)$ whenever the vertices $x, y, z$ from the same part ($U$ or
$V$) of the bipartition, ordered $x < y < z$, together with some
vertex $v$ from the other part of the bipartition ($V$ or $U$),
induce a forbidden pattern in $\F$. There is also an arc from
$(x,y)$ to $(u,v)$ and an arc from $(v,u)$ to $(y,x)$ whenever the
vertices $x, y$ from the same part, ordered as $x<y$, together
with some vertices $u, v$ from the other part, ordered as $v<u$,
induce a pattern in $\F$.

We say that a pair $(x,y)$ {\em dominates} $(x',y')$ and we write
$(x,y) \rightarrow (x',y')$ if there is an arc from $(x,y)$ to
$(x',y')$ in $H^+$.

A {\em circuit\/} in a subset $D$ of $H^+$ is a sequence of pairs
$(x_0,x_1),(x_1,x_2),\dots,\allowbreak (x_{n-1},x_n),\allowbreak
(x_n,x_0)$, $n\ge1$, that all belong to $D$. Observe that
$x_0,x_1,\dots,x_n$ belong to the same bipartition part of $V(H)$.

\bigskip

\noindent
{\sc Ordering with bipartite forbidden $4$-patterns, $\BORD_4$}\\
\noindent
{\sc Input:} A bigraph $H=(U,V)$ and a set $\F \subseteq \B_4$ of
bipartite forbidden patterns on four vertices\\
\noindent {\sc Output:} And ordering of the vertices in $U$ and an
ordering of the vertices in $V$ that is a $\F$-free ordering or
report that there is no such ordering.

\bigskip

\noindent {\sc Algorithm for $\BORD_4$}
\begin{enumerate}
\item If a strong component $S$ of $H^+$ contains a circuit then
report that no solution exists and exit.
      Otherwise, remove $\emptyset$, $\{11',12',21',22'\}, \{11',12',13' \}$ and
$\{11',21',31'\}$ from $\F$. If $\F$ is empty after this step,
then return any ordering of vertices of $H$ and stop. \item Set
$D$ to be the empty set. \item Choose a strong component $S$ of
$H^+$ that is green with respect to $D$. The choice is made
according to the following rules.

a) If $\F$ contains one of the forbidden patterns
$\{11',12',21'\},\{12',21',22'\},$ \\ $
\{11',12',22'\},\{11',21',22'\}$ then priority is given to a
component $S$ containing $(x,y)$ where $x,y$ have a common
neighbor in $H$. If there is a choice then it is preferred $S$ to
be a trivial component. Subject to these preferences, if there are
several candidates, then priority is given to the ones that are
sink components in $H^+$.

b) If $\F$ contains one of the forbidden patterns
$\{11'\},\{22'\}),\{12'\}, \{21'\}$ then priority is given to a
component $S$ containing $(x,y)$ where $x,y$ have a common
non-neighbor in $H$. If there is a choice then it is preferred $S$
to be a trivial component. Subject to these preferences, if there
are several candidates, then priority is given to the ones that
are sink components in $H^+$.

\item If by adding $S$ into $D$ we do not close a circuit, then we
add $S$ into $D$ and discard $\overline{S}$. Otherwise we add
$\overline{S}$ and its outsection (all vertices in $H^+$ that are
reachable from $\overline{S}$) into $D$ and discard $S$ and its
insection (the vertices that can reach $S$). Return to Step 3 if
there are some strong components of $H^+$ left.

\item For every $(x,y) \in D$, place $x$ before $y$ in the final
ordering.
\end{enumerate}

A polynomial-time solution to $\BORD_4$ is implicit in the
following fact, the main result of this section.

\begin{theorem}\label{main-bipartite}
Let $\F \subseteq \B_4$ and let $H^+$ be the constraint digraph of $H$
with respect to $\F$. Then $H$ has a $\F$-free ordering of its
parts if and only if no strong component of $H^+$ contains a
circuit.
\end{theorem}

The proof, and the correctness of the algorithm will also be given
in the next section.

\begin{corollary}
Each problem $\BORD(\F)$ with $\F \subseteq \F_4$ can be solved in
polynomial time.
\end{corollary}

\section{The Correctness of the Algorithms}

The validity of the first step of the algorithm is justified by
Lemmas \ref{if-circuit} and \ref{lem:empty pattern}. Thus we may
assume from now on that no strong component of $H^+$ contains a
circuit, that $\emptyset$ and $\{12,13,23\}$ are not in $\F$ and
that $\F\ne\emptyset$.

Observe that a strong component $S$ of $H^+$ contains a circuit if
and only if $\overline{S}$ contains a circuit. Moreover, $S \cap
\overline{S} = \emptyset$ as otherwise for every $(x_0,x_1) \in S
\cap \overline{S}$, we have $(x_1,x_0) \in  S \cap \overline{S}$
and hence there would be a circuit $(x_0,x_1),(x_1,x_0)$ in $S$.

We first need the following lemma about the structure of strong
components of $H^+$.

\begin{lemma}\label{component-structures}
If $\F$ has only one element and that element is one of
$\{12,23\}, \{13\}$, then $H^+$ is symmetric (it is just a graph).
\end{lemma}
\begin{proof} We prove then lemma when $\F=\{\{12,23\}\}$ and the proof for
the other case is obtained by applying the arguments in the
complement of $H$. Suppose $(x,y) \rightarrow (z,y)$ according to
$\{12,23\}$. Thus by definition of $H^+$ we have $xy,yz \in E(H)$
and $xz \not\in E(H)$. By considering the order $z<y<x$, we
conclude that $(z,y) \rightarrow (x,y)$. Hence, $(x,y)(z,y)$ is a
symmetric arc. Now suppose $(x,y) \rightarrow (x,z)$ according to
$\{12,23\}$. Thus  by definition of $H^+$ we have $xz,xy \in E(H)$
and $yz \not\in E(H)$ and hence $(x,z) \rightarrow (x,y)$ implying
that $(x,y)(x,z)$ is a symmetric arc.\qed \end{proof}

\begin{claim}\label{cycle}
Let $C=w_1w_2\dots w_kw_1$ be an induced cycle of length $k \ge 4$
in $H$ and let $\{13,23\} \in \F$. Then the strong component of
$H^+$ containing $(w_1,w_2)$ contains a circuit.
\end{claim}

\begin{proof} To prove this claim, we first observe that $(w_i,w_{i+1})
\rightarrow (w_i,w_{i+2}) \rightarrow (w_{i+1},w_{i+2})$ for $0
\le i \le k-1$ (where all indices are taken modulo $k$). Hence,
$(w_1,w_2)$, $(w_2,w_3)$, \dots, $(w_{k-1},w_k)$, $(w_k,w_1)$ are
all in the same strong component of $H^+$, and thus there exists a
circuit in this strong component. This proves the claim. \qed
\end{proof}

We will use the following notation. For $x\in V(H)$, we let $N(x)$
be the set of all neighbors of $x$ in $H$.

\begin{claim}\label{sink-component}
If $(x,y)$ does not dominate any pair in $H^+$, then one of the
following happens:
\begin{itemize}
\item[1)] $N(x)\setminus\{y\} \subseteq N(y)\setminus\{x\}$,
\item[2)] $N(y)\setminus\{x\} \subseteq N(x)\setminus\{y\}$, or
\item[3)] $(y,x)$ does not dominate any pair in $H^+$.
\end{itemize}
\end{claim}
\begin{proof} We show the argument when $xy$ is an edge of $H$ and the case
$xy \not\in E(H)$ follows by applying the argument in the
complement of $H$. Suppose none of conditions 1) and 2) happens.
Then there exists $z \in N(x)\setminus (N(y)\cup\{y\})$ and there
exists $w \in N(y)\setminus (N(x)\cup\{x\})$. Now $\{12,13\}$ is
not in $\F$ as otherwise $(x,y)$ would dominate the pair $(z,y)$
in $H^+$. Similarly, none of $\{12,23\}$ and $\{13,23\}$ is a
forbidden pattern in $\F$ as otherwise $(x,y)$ would dominate the
pair $(x,z)$ or the pair $(x,w)$ (respectively). Thus, we may
assume that $\F\subseteq \{\{12\},\{23\},\{13\}\}$. If $\{12\} \in
\F$, then there is no vertex $v \in V(H)$ outside $N(x) \cup N(y)$
as otherwise $(x,y) \rightarrow (v,y)$, a contradiction. Thus
$(y,x)$ does not dominate any pair in $H^+$ since for every $v\ne
x,y$, the subgraph induced on $v,x,y$ contains at least two edges.
Hence 3) holds. Similar argument works if $\{23\} \in \F$.
Finally, if $\F = \{\{13\}\}$, then by the assumption and by Lemma
\ref{component-structures}, none of $(x,y)$ and $(y,x)$ dominates
a pair in $H^+$ and hence 3) holds. \qed \end{proof}

\begin{lemma}\label{correctness-for-3}
The Algorithm for $\ORD_3$ does not create a circuit in $D$.
\end{lemma}

The idea of the proof is to show that a circuit that is created
first, and also has a minimum length amongst circuits created at
the same time, forces a special modular structure on the graph
$H$, and this eventually implies a contradiction.

\hspace{1mm}

\begin{proof} Suppose that by adding a green component $S$ into $D$ we close
a circuit $C: (x_0,x_1),(x_1,x_2),\dots,\allowbreak
(x_{n-1},x_n),(x_n,x_0)$ in $D \cup S$ for the first time. We may
assume that $n$ is minimum and $(x_n,x_0) \in S$. We also assume
that $\F$ contains one of the patterns $\{13,23\}, \{12,13\},
\allowbreak \{12,23\}$. This follows from the fact that we may
apply our arguments for the complementary forbidden patterns in
the complement of $H$. This choice will enable us to concentrate
on Case (a) in Step 3 of the algorithm.

The proof is divided into the following four cases:

(A) $S$ is a trivial component and $x_nx_0$ is an edge.

(B) $S$ is a trivial component and $x_nx_0$ is not an edge.

(C1) $S$ is a non-trivial component and $\{13,23\} \in \F$ (or, by
symmetry, $\{12,13\} \in \F$).

(C2) $S$ is a non-trivial component and $\{12,23\} \in \F$.

The proof will show that when the first circuit $C:
(x_0,x_1),(x_1,x_2),\dots,\allowbreak (x_{n-1},x_n),(x_n,x_0)$ is
created, then the shortest circuit  created at this time has
$x_0,x_1,\dots,x_n$ either induce a clique or induce an
independent set with special adjacencies to the other vertices of
$H$. This will imply that $\overline{S}$ and its outsection can be
added into $D$ without creating a circuit.

\medskip

{\bf A. Suppose that $(x_n,x_0)$ forms a trivial strong component
$S$ in $H^+$, and $x_nx_0$ is an edge of $H$.}

First suppose $(x_n,x_0)$ does not dominate any pair in $H^+$. Now
according to the algorithm each pair $(x_i,x_{i+1})$, $0 \le i \le
n$ is also in a trivial strong component and it does not dominate
any other pair in $H^+$ and $x_ix_{i+1}$ is an edge of $H$. We
show that $(x_0,x_n)$ is also in a trivial component and it is
green. For a contradiction suppose $(x_0,x_n)$ dominates a pair in
$H^+$. Since $(x_n,x_0)$ does not dominate any pair in $H^+$ and
by assumption $(x_0,x_n)$ dominates a pair in $H^+$, Claim
\ref{sink-component} implies that either $N(x_n)\setminus\{x_0\}
\subset N(x_0)\setminus\{x_n\}$ or $N(x_0)\setminus\{x_n\} \subset
N(x_n)  \setminus \{x_0\}$. W.l.o.g assume $N(x_n)\setminus\{x_0\}
\subset N(x_0)\setminus\{x_n\}$. Now $N(x_0)\setminus\{x_n\}
\not\subseteq N(x_n)\setminus\{x_0\}$ as otherwise $(x_0,x_n)$
does not dominate any pair in $H^+$. Now these together with Lemma
\ref{component-structures}  imply that $\{12,13\} \in \F$. Since
$(x_0,x_1)$ does not dominate any pair in $H^+$, we conclude that
$N(x_0)\setminus \{x_1\} \subseteq N(x_1)\setminus \{x_0\}$ and by
continuing this argument we conclude that
$N(x_{n-1})\setminus\{x_n\} \subseteq N(x_n)\setminus\{x_{n-1}\}$.
Therefore $N(x_0)\setminus\{x_n\} \subseteq
N(x_n)\setminus\{x_0\}$, a contradiction.

Therefore $S= \{(x_0,x_n)\}$ is a green strong component. If by
adding $(x_0,x_n)$ into $D$ we close a circuit $C_1 :
(x_n,y_1),(y_1,y_2),\dots,(y_m,x_0)$ then there would be an
earlier circuit
$$(x_0,x_1),(x_1,x_2),\dots,(x_{n-1},x_n),(x_{n},y_1),(y_1,y_2),\dots,(y_m,x_0),$$
a contradiction (Note that since $(x_0,x_n),(x_n,x_0)$ are
singleton, all the pairs in $C,C_1$ apart from
$(x_0,x_n),(x_n,x_0)$ are already in $D$ ).

\noindent \textit{Now we continue by assuming $(x_n,x_0)$ is a
singleton component and $x_nx_0$ is an edge and $(x_n,x_0)$
dominates a pair in $H^+$.}

According to the Algorithm, $(x_i,x_{i+1})$ is also in a trivial
component and $x_ix_{i+1}$, $0 \le i \le n$ is an edge. Moreover,
$(x_0,x_n)$ is in a trivial component and it dominates a pair in
$H^+$ as otherwise it should have been considered before
$(x_n,x_0)$.

\begin{claim}\label{clique}
$x_0,x_1,\dots,x_n$ induce a clique in $H$.
\end{claim}
\begin{proof} Consider the edge $x_ix_j$, $i \ne j-1$ in $H$ (note that such
an edge exists since $x_ix_{i+1}$, $0 \le i \le n$ is an edge).
Whenever $\{12,23\}$ is a forbidden pattern in $\F$, $x_{i-1}x_j$
is an edge of $H$ as otherwise $(x_{i-1},x_{i})$ dominates
$(x_{j},x_i)$ and hence $(x_j,x_i) \in D \cup S$, implying a
shorter circuit
$(x_i,x_{i+1}),(x_{i+1},x_{i+2}),\dots,(x_{j-1},x_j),(x_j,x_i)$ in
$D \cup S$ (the indices are  modul $n$). Whenever $\{13,23\}$ is a
forbidden pattern, $x_{i-1}x_j$ is an edge of $H$ as otherwise
$(x_{i-1},x_i)$ dominates $(x_{i-1},x_j)$ and $(x_{i-1},x_j)
\rightarrow (x_i,x_j)$ and hence both $(x_{i-1},x_j),(x_i,x_j) \in
D \cup S$, implying a shorter circuit in $D \cup S$. By applying
this argument when one of the $\{12,23\},\{13,23\}$ is in $\F$ we
conclude that $x_rx_j$ is an edge for every $r \ne j$. Therefore
$x_0,x_1,\dots,x_n$ induce a clique in $H$. Now suppose $x_ix_j$
is an edge for $j \ne i+1$ (note that such an edge exists since
$x_ix_{i+1}$, $0 \le i \le n$ is an edge). Whenever $\{12,13\}$ is
a forbidden pattern, $x_{i+1}x_j$ is an edge as otherwise
$(x_i,x_{i+1})$ dominates $(x_{j},x_{i+1})$ and hence we obtain a
shorter circuit in $D \cup S$. Since $x_ix_{i+1},x_{i+1}x_{i+2}$
are edges of $H$, by similar argument we conclude that $x_rx_j$ is
an edge for every $r \ne j$. Therefore $x_0,x_1,\dots,x_n$ induce
a clique in $H$. \qed \end{proof}

Note that since $(x_n,x_0)$ is in a trivial component and
$(x_n,x_0)$ dominates a pair in $H^+$, by Lemma
\ref{component-structures}  either $\{13,23\}$ or $\{12,13\}$ is
in $\F$ and $\{12,23\} \not\in \F$. Therefore we consider the case
$\{13,23\} \in \F$ and the argument for $\{12,13\} \in \F$ is
followed by symmetry.

\noindent\textit{ $\{13,23\}$ is a forbidden pattern in $\F$. }

First suppose $(x_0,x_n)$ dominates a pair according to
$\{13,23\}$. Thus there exists $p_n$ such that $x_np_n$ is an
edges of $H$ and $x_0p_n$ is not an edge of $H$. Therefore there
exists a smallest index $1 \le i \le n$ such that $p_nx_i$ is an
edge and $p_nx_{i-1}$ is not an edge. Now $(x_{i-1},x_{i})
\rightarrow (x_i,p_n)$ and $(x_i,p_n) \rightarrow (x_n,p_n)$. This
would imply that $(p_n,x_n) \in D$.

If $i=1$ then both $(x_0,p_n),(x_n,p_n)$ are in $D$ and hence
$(x_0,x_n)$ does not dominates a pair outside $D$ according to
$\{13,23\}$. Therefore we may assume $i > 1$ and $(x_0,p_n)$
dominates some pair $(w,p_n)$. We must have $wx_i \in E(H)$ as
otherwise $(x_0,w),(w,p_n),(p_n,x_i),(x_i,x_0)$ is a circuit in a
strong component of $H^+$.  Now when $x_{i-1}w$ is not an edge of
$H$, $(x_{i-1},x_i) \rightarrow (x_{i-1},w)$ and hence
$(x_{i-1},w) \rightarrow (x_0,w) \rightarrow (x_0,p_n) \in D$
implying $(x_0,x_n)$ is also green. Thus we may assume that
$wx_{i-1}$ is an edge. Now $(x_{i-1},p_n) \rightarrow (w,p_n)$,
implying that $(w,p_n) \in D$. These imply that the only pair that
$(x_0,x_n)$ may dominate outside $D$ is $(x_0,p_n)$. However if
there is no other forbidden pattern in $\F$, by adding $(x_0,p_n)$
into $D$ we do not close a circuit. Otherwise such a circuit
comprises of the pairs $(p_n,y_1),(y_1,y_2),\dots,(y_m,x_0)$ in
$D$. Now $$(x_0,x_1),(x_1,x_2),\dots,(x_{n-1},x_n),(x_n,p_n),
p_n,y_1),(y_1,y_2),\dots,(y_m,x_0)$$ yield an earlier circuit in
$D$, a contradiction. Thus by adding $(x_0,p_n)$ into $D$ and then
adding $(x_0,x_n)$ into $D$ we do not close a circuit into $D$ as
otherwise we conclude that there would be an earlier circuit in
$D$.

Now we continue by assuming there is another forbidden pattern in
$\F$. We show that in this case there is no pair dominated by
$(x_0,x_n)$ outside $D$ according to this forbidden pattern. First
suppose $\{12\}$ is also a forbidden pattern in $\F$ and
$(x_0,x_n)$ dominates some pair $(q_n,x_n)$ according to $\{12\}$,
where $q_n$ is not adjacent to any of $x_0,x_n$. We note that
$(x_n,x_0)$ dominates $(q_n,x_0)$ and since $(x_n,x_0)$ is
singleton, $(q_n,x_0) \in D$. We show that $(q_n,x_n) \in D$.
Observe that we may assume that $x_{n-1}q_n \in E(H)$ as otherwise
$(x_{n-1},x_n) \rightarrow (q_n,x_n)$ and hence $(q_n,x_n) \in D$
and we are done. Now $x_{q-2}q_n$ must be an edge of $H$ as
otherwise $(x_{n-2},x_{n-1})$ dominates $(x_{n-2},q_n)$ and we
obtain an earlier circuit
$(x_0,x_1),(x_1,x_2),\dots,(x_{n-3},x_{n-2}),(x_{n-2},q_n),(q_n,x_0)$
in $D$. By continuing this argument we conclude that $q_nx_0$ must
be an edge as otherwise we obtain an earlier circuit in $D$, a
contradiction. Therefore $(x_0,x_n)$ is also green and hence by
adding $(x_0,x_n)$ into $D$ we do not close a circuit as otherwise
there would be an earlier circuit in $D$.

Analogously we conclude that if $\{23\}$ is a forbidden pattern
then we can add $(x_0,x_n)$ into $D$ without creating a circuit.
Note that $(x_0,x_n)$ does not dominate any pair using the
forbidden pattern $\{13\}$.

The last remaining case is when $\{12,13\}$ is also a forbidden
pattern. In this case $(x_0,x_n)$ does not dominate any $(z,x_n)$
in $H^+ \setminus (D \cup \overline{D})$ according to $\{12,13\}$
($zx_0 \in E(H)$, $zx_n \not\in E(H)$) as otherwise $(x_n,x_0)$
dominates $(x_n,z)$, and hence $(z,x_n) \in \overline{D}$, a
contradiction.

{\bf B. Suppose that $(x_n,x_0)$ forms a trivial strong component
$S$ in $H^+$, and $x_nx_0$ is not an edge of $H$.}

According to the algorithm if $x_nx_0$ is not an edge then
$(x_0,x_n)$ does not dominate any pair in $H^+$ according to
$\{12,13\},\{13,23\}$. Moreover if there exists $q$ such that
$qx_0 \in E(H)$ and $qx_n \not\in E(H)$ then according to the
algorithm none of $\{12\},\{23\}$ is in $\F$. Otherwise since
$qx_0$ is an edge then one of $(q,x_0)$ or $(x_0,q)$ should have
been considered earlier and hence one of them is in $D$. This
would imply that $(x_n,x_0)$ or $(x_0,x_n)$ is also dominated by
one $(q,x_0)$, $(x_0,q)$ and hence one of the
$(x_0,x_n),(x_n,x_0)$ is in $D$ a contradiction. Similarly if
there exist $q'$ such that $q'x_n$ is an edge and $q'x_0$ is not
an edge we conclude that none of $\{12\},\{23\}$ is in $\F$. Thus
either $(x_0,x_n)$ is green and in this case by adding $(x_0,x_n)$
into $D$ we don't create a circuit as otherwise there would be an
earlier circuit in $D$ or $\{13\} \in \F$ which implies that
$(x_n,x_0)$ is in a non-trivial component, a contradiction.

{\bf C. Suppose that $(x_n,x_0)$ belongs to a non-trivial strong
component $S$ in $H^+$.}

{\bf Case 1.} $\{13,23\} \in \F$.

\noindent\textit{ First suppose $x_0x_n$ is an edge of $H$. }
Since $(x_n,x_0)$ is in a non-trivial component, there must be
some other forbidden pattern in $\F$. We first consider the case
that one of the $\{12\},\{23\},\{13\}$ is also in $\F$. As a
consequence $x_0,x_1,\dots,x_n$ induce a clique (To see this:
$x_1$ must be adjacent to both $x_0,x_n$ otherwise $(x_n,x_0)
\rightarrow (x_1,x_0)$ for $\{12\} \in \F$ or $(x_0,x_1)
\rightarrow (x_0,x_n)$ for $\{23\} \in \F$ or $(x_0,x_1)
\rightarrow (x_n,x_1)$ for $\{13\} \in \F$ and hence in any case
we obtain a shorter circuit. Now it is easy to see that (using the
argument in Claim \ref{clique}) both $x_0x_1,x_nx_1$ must be edges
of $H$ and by continuing this argument $x_0,x_1,\dots,x_n$ induce
a clique in $H$).

Since $(x_n,x_0)$ is in a non-trivial component, there exists some
pair $(u,v) \in S$ that dominates $(x_n,x_0)$. First suppose
$(u,v) \rightarrow (x_n,x_0)$ according to $\{13,23 \}$. In this
case there exists $q_n$ such that $q_nx_n \in E(H)$ and $x_0q_n
\not\in E(H)$ and $(q_n,x_0) \rightarrow (x_n,x_0)$. Now $q_nx_j
\not\in E(H)$, $j \ne 0,n$ as otherwise $(q_n,x_0) \rightarrow
(x_j,x_0)$ and hence we have a shorter circuit
$(x_0,x_1),(x_1,x_2),\dots,(x_{j-1},x_j),(x_j,x_0)$ in $D \cup S$.
Observe that $(x_{n-1},x_n) \rightarrow (x_{n-1},q_n) \rightarrow
(x_n,q_n)$. Now if $\{13\} \in \F$ then $(q_n,x_0) \rightarrow
(q_n,x_{n-1})$ and hence we have a shorter circuit in $D \cup S$.
If $\{12\} \in \F)$ then $(x_{n-2},x_{n-1}) \rightarrow
(q_n,x_{n-1})$, and again there would be a shorter circuit in $D
\cup S$. Finally when $\{23\} \in \F$ we have $(x_0,x_1)
\rightarrow (x_0,q_n)$ while $(q_n,x_0) \in D \cup S$, a
contradiction. Therefore we may assume that $(x,y) \in S$
dominates $(x_n,x_0)$ according to one of the
$\{12\},\{23\},\{13\}$. First suppose $(x_n,q_n)$ dominates
$(x_n,x_0)$ according to $\{12\}$. Note that $q_nx_0,q_nx_n
\not\in E(H)$. Now $(x_n,x_0) \rightarrow (q_n,x_0)$. We show that
$q_n$ is not adjacent to any of $x_0,x_1,\dots,x_n$. For a
contradiction suppose there exists $x_j$ such that $x_{j}q_n
\not\in E(H)$ and $x_{j+1}q_n \in E(H)$. Now $(x_{j},x_{j+1})
\rightarrow (x_{j},q_n)$ according to $\{13,23\} \in \F$ and hence
we obtain a shorter circuit
$(x_0,x_1),(x_1,x_2),\dots,(x_{j-1},x_j),(x_j,q_n),(q_n,x_0)$ in
$D \cup S$. This implies that $(x_{n-1},x_n) \rightarrow
(q_n,x_n)$. However $(x_n,q_n) \rightarrow (x_n,x_1)$ according to
$\{12\}$ and hence we obatin a shorter circuit in $S \cup D$.
Analogously if some pair $(x,y)$ dominates $(x_n,x_0)$ according
to $\{23\}$ we arrive at a contradiction. Note that no pair
$(x,y)$ dominates $(x_n,x_0)$ according to $\{13\}$.

Thus we continue by assuming that none of the
$\{12\},\{23\},\{13\}$ belongs to $\F$ but one of the
$\{12,13\},\{12,23\} \in \F$.

\noindent\textit{ First, suppose $\{12,23\} \in \F$.} Observe that
$x_ix_{i+1} \in E(H)$, $0 \le i \le n$ (otherwise according to the
priorities of the pairs there must be some vertex $q_i$ such that
$q_ix_i,q_ix_{i+1} \in E(H)$ and $(x_i,q_i) \in D$ and
$(x_i,q_{i}) \rightarrow (x_i,x_{i+1}) \rightarrow (q_i,x_{i+1})$.
However $(x_i,q_i) \rightarrow (x_{i+1},q_i)$ according to
$\{12,23\}$ a contradiction). Thus by Claim \ref{clique}
$x_0,x_1,\dots,x_n$ induce a clique. Suppose $(u,v) \in S$, and
$(u,v) \rightarrow (x_n,x_0)$ according to $\{13,23\}$. This means
there is $q_n$ such that $q_nx_0 \not\in E(H)$ and $q_nx_n \in
E(H)$ and $(q_n,x_0) \rightarrow (x_n,x_0)$. Now according to
$\{13,23 \}$, $(x_n,x_0) \rightarrow (q_n,x_0) \rightarrow
(q_n,x_n)$ and hence both $(x_n,q_n)$ and $(q_n,x_n)$ are in $ D
\cup S$, a contradiction. Thus we may assume there is some pair
$(x,y) \in S$ dominates $(x_n,x_0)$ according to $\{12,23\}$.
Either there exists $q_0$ such that $x_0q_0$ is an edge, $x_nq_0
\not\in E(H)$ and $(q_0,x_0) \in S$ or there exists $q_n$ such
that $q_nx_n \in E(H)$, $q_nx_0 \not\in E(H)$ and $(x_n,q_n)
\rightarrow (x_n,x_0)$. If the first case happens then according
to $\{13,23\}$, $(q_0,x_0) \rightarrow (q_0,x_n) \rightarrow
(x_0,x_n)$, implying that $(x_0,x_n) \in S \cup D$. In the former
case $x_{n-1}q_n \in E(H)$ as otherwise $(x_{n-1},x_n) \rightarrow
(x_n,q_n)$ a contradiction. However there exists some $j$ such
that $x_jq_n \not\in E(H)$ and $x_{j+1}q_n \in E(H)$. Now
$(x_{j},x_{j+1}) \rightarrow (q_n,x_{j+1})$ according to
$\{12,23\}$ and $(x_j,x_{j+1}) \rightarrow (x_j,q_n) \rightarrow
(x_{j+1},q_n)$ according to $\{13,23\}$. Therefore
$(x_{j+1},q_n),(q_n,x_{j+1}) \in D \cup S$, a contradiction.

\noindent\textit{ Second, suppose $\{12,13\} \in \F$.} We show
that $x_ix_{i+1}$, $0 \le i \le n$ is not an edge. For a
contradiction  suppose $x_ix_{i+1}$ is not an edge. Now there
exists $q_i$ such that $q_ix_i,q_ix_{i+1} \in E(H)$ and
$(q_i,x_{i+1}) \rightarrow (x_i,x_{i+1})$ according to $\{12,13\}$
or $(x_i,q_i) \rightarrow (x_i,x_{i+1})$ according to $\{13,23\}$.
In any case we have $(x_i,q_{i}),(x_i,x_{i+1})$, \\
$(q_i,x_{i+1}) \in D \cup S$. Now $q_ix_j \not\in E(H)$, $j \ne
i,i+1$ as otherwise when $x_{i+1}x_{j} \not\in E(H)$,
$(q_i,x_{i+1}) \rightarrow (x_{j},x_{i+1})$ according to
$\{12,13\}$ and when $x_{i+1}x_{j}$ is an edge $x_ix_j \not\in
E(H)$ otherwise $(x_{i},x_{i+1}) \rightarrow (x_{j},x_{i+1})$ and
hence $(x_i,q_i) \rightarrow (x_i,x_j)$, a shorter circuit. Note
that when $x_jx_{j+1} \not\in E(H)$, $j \ne i$ then $q_iq_j
\not\in E(H)$ as otherwise we obtain an induced $C_4$. Now we
obtain an induced cycle of length more than 3 using $x_nx_0$ and
$x_i,q_i,x_{i+1}$'s, a contradiction.

Therefore $x_0,x_1,\dots,x_n$ induce a clique according to Claim
\ref{clique}. First suppose $(u,v) \rightarrow (x_n,x_0)$
according to $\{13,23 \}$. In this case there exists $q_n$ such
that $q_nx_n \in E(H)$ and $x_0q_n \not\in E(H)$ and $(q_n,x_0)
\rightarrow (x_n,x_0)$. Now $q_nx_j \not\in E(H)$, $j \ne 0,n$ as
otherwise $(q_n,x_0) \rightarrow (x_j,x_0)$ and hence we obtain a
shorter circuit $(x_0,x_1),(x_1,x_2),\dots,$\\
$(x_{j-1},x_j),(x_j,x_0)$ in $D \cup S$. However $(x_n,x_0)
\rightarrow (q_n,x_0) \rightarrow (q_n,x_n)$ and hence both
$(x_n,q_n)$ and $(q_n,x_n)$ are in $ D \cup S$, according to
$\{12,13\}$, a contradiction. Therefore we assume that $(x_n,x_0)$
is dominated by a pair according to $\{12,13\}$ and analogously we
arrive at a contradiction.

\noindent\textit{ Second, suppose $x_0x_n \not\in E(H)$.}

\begin{claim}\label{only-{13,23}}
$\{13,23\}$ is not the only forbidden pattern in $\F$.
\end{claim}
\begin{proof} For contradiction suppose $\{13,23\}$ is the only forbidden
pattern. Note that by assumption there exists at least one $0 \le
r \le n$ such that $x_rx_{r+1}$ is not an edge of $H$ (in
particular $r=n$, the indexes are modul $n+1$). Now according to
the rules of the algorithm the assumption that $\{13,23\}$ is the
only forbidden pattern, if $x_ix_{i+1} \not\in E(H)$, $0 \le i \le
n$ there must be a vertex $p_i$ such that $x_ip_i,x_ip_{i+1} \in
E(H)$ and $(x_i,p_i) \in S \cup D$ dominates $(x_i,x_{i+1})$.
Moreover $(x_i,x_{i+1}) \rightarrow (p_i,x_{i+1})$ and hence
$(p_i,x_{i+1}) \in S \cup D$. We may assume circuit $C$ has the
minimum number of pairs $(x_i,x_{i+1})$ that $x_ix_{i+1}$ is not
an edge.

{\bf Observation :} If $x_jx_{j+1}, x_{j+1}x_{j+2}$ are edges of
$H$ then $x_jx_{j+2}$ is also an edge of $H$ as otherwise
$(x_j,x_{j+1}) \rightarrow (x_j,x_{j+2})$ and hence $(x_j,x_{j+2})
\in S \cup D$, implying a shorter circuit. If $x_jx_{j+1}$ is not
an edge of $H$ then for every $j' \ne j,j+1$ at most one of the
$x_jx_{j'},x_{j+1}x_{j'}$ is an edge of $H$ as otherwise
$(x_j,x_{j+1}) \rightarrow (x_{j'},x_{j+1})$ and hence we have a
shorter circuit in $S \cup D$.

\noindent\textit{ First, suppose $x_{r+1}x_{r+2}$ is an edge of
$H$}. By observation above $x_rx_{r+2} \not\in E(H)$ and hence
$p_rx_{r+2} \not\in E(H)$ as otherwise $(x_r,p_r) \rightarrow
(x_r,x_{r+2})$ and hence $(x_r,x_{r+2}) \in S \cup D$, a shorter
circuit in $S \cup D$. If $x_{r+2}x_{r+3}$ is also an edge of $H$
then by Observation above $x_{r+1}x_{r+3} \in E(H)$ and hence
$x_rx_{r+3} \not\in E(H)$. Now $p_rx_{r+3} \not\in E(H)$ as
otherwise $(x_r,p_r) \rightarrow (x_r,x_{i+r})$, implying a
shorter circuit in $S \cup D$.

\noindent\textit{ Second, let $j$ be the first index after $r$ (in
the clockwise direction) such that $x_jx_{j+1}$ is not an edge of
$H$.} By the above Observation we may assume that $x_{r+1}x_j$,
$r+1 \ne j$ is an edge of $H$ and none of
$x_rx_j,x_rx_{j+1},x_{r+1}x_{j+1}$ is an edge of $H$. Now
$p_{j}x_r \not\in E(H)$ as otherwise $(x_{j},p_j) \rightarrow
(x_j,x_r)$ and hence $(x_j,x_r) \in S \cup D$, implying a shorter
circuit in $S \cup D$. Now $x_{j+1}x_r \not\in E(H)$ as otherwise
$(p_{j},x_{j+1}) \rightarrow (p_{j},x_r) \rightarrow
(x_{j+1},x_r)$ and hence $(x_{j+1},x_r) \in S \cup D$, a shorter
circuit in $S \cup D$. Now $p_rx_{j+1} \not\in E(H)$ as otherwise
$(x_r,p_r) \rightarrow (x_r,x_{j+1})$ and hence we obtain a
shorter circuit
$(x_0,x_1),(x_1,x_2),\dots,(x_{r-1},x_r),(x_r,x_j),\dots,
(x_n,x_0) \in S \cup D$. Similarly $p_rx_j \not\in E(H)$. as
otherwise we obtain a shorter circuit in $S \cup D$. Moreover
$p_rp_{j} \not\in E(H)$ as otherwise $(x_r,p_r) \rightarrow
(x_r,p_j)$ and by replacing $x_j$ with $p_{r+1}$ we obtain circuit
$C_1=(x_0,x_1),\dots,(x_{r-1},x_r),(x_r,p_j),(p_{j},x_{j+1}),(x_{j+1},x_{r+2}),$
\\ $\dots,(x_n,x_0)$, and since $p_{r+1}x_{r+2}$ is an edge of $H$,
$C_1$ contradicts the our assumption about $C$ (if $j \ne r+1$
then $(p_r,x_{r+1}) \rightarrow (p_r,x_j) \rightarrow
(x_{r+1},x_j)$ and hence $(x_{r+1},x_j) \in S \cup D$, a shorter
circuit in $S \cup D$).

By applying two above arguments we conclude that $p_r$ is not
adjacent to any $x_j$, $j \ne r,r+1$ and $p_r$ is not adjacent to
any $p_j$. Therefore we obtain an induced cycle
$x_r,p_r,x_{r+1},x_j,p_j,x_{j+1},\dots,x_r$ of length more than
$3$, and by Claim \ref{cycle} we conclude that there exists a
circuit in a strong component of $H^+$. \qed \end{proof}

\begin{claim}\label{independent}
If one of the patterns $\{12\},\{23\},\{13\}$ is in $\F$, then
$x_0,x_1,\dots,x_n$ induce an independent set.
\end{claim}

\begin{proof} For a contradiction suppose $x_ix_j \in E(H)$. Now $x_{i-1}x_i
\in E(H)$ as otherwise $(x_{i-1},x_i) \rightarrow (x_j,x_i)$ when
$x_jx_{i-1} \in E(H)$ according to $\{13,23\}$ and when
$x_{i-1}x_j \not\in E(H)$ then $(x_{i-1},x_i) \rightarrow
(x_j,x_i)$ when  $\{23\} \in \F$ or $(x_{i-1},x_i) \rightarrow
(x_{i-1},x_j)$ when $\{23\} \in \F$) or $(x_{i-1},x_i) \rightarrow
(x_{i-1},x_j)$ when $\{23\} \in \F$). In any case we have a
shorter circuit. Therefore $x_{i-1}x_i \in E(H)$ and by continuing
this argument we conclude that $x_0x_n \in E(H)$, contradicting
our assumption. \qed \end{proof}

Since $x_0x_n$ is not an edge up to symmetry and using the same
argument in the complement of $H$ in the remaining we just
consider the cases  $\{13\} \in \F$ and $\{12,23\} \in \F$.

\noindent\textit{We first consider $\{13\} \in \F$.} Suppose
$(x_n,x_0)$ is dominated by some pair $(x,y) \in S$ according to
forbidden pattern $\{13,23\}$. This means that there exists some
$q_n$ such that $(x_n,q_n) \in S$ dominates $(x_n,x_0)$,
$q_nx_0,q_nx_n \in E(H)$. Now $x_{n-1}q_n \in E(H)$ as otherwise
according to forbidden pattern $\{13\}$, $(x_{n-1},x_n)
\rightarrow (x_{n-1},q_n)$ and $(x_{n-1},q_n) \rightarrow
(x_{n-1},x_0)$ a shorter circuit in $D \cup S$. However according
to $\{13,23\}$, $(x_n,q_n) \rightarrow (x_n,x_{n-1})$, a
contradiction. Therefore $(x_n,x_0)$ is dominated by a pair $(x,y)
\in S$ according to pattern $\{13\}$. This means there exists
$q_0$ such that $q_0x_0 \in E(H)$, $q_0x_n \not\in E(H)$ and
$(x_n,q_0) \in S$ dominates $(x_n,x_0)$. Consider an edge $x_iq_i$
for some $0 \le i \le n$. By applying similar argument in the
beginning of the case, $q_ix_{i+1} \not\in E(H)$. We show that
$q_i$ is not adjacent to any other $x_j$, $j \ne i,i+1$. Otherwise

$(x_i,x_{i+1}) \rightarrow (q_i,x_{i+1})$ and  $(q_i,x_{i+1})
\rightarrow (x_j,x_{i+1})$ and hence we obtain a shorter circuit.
Moreover we show that $x_i$ can be replaced by $x_i$ and obtain a
circuit of length $n$ as follows: $(x_{i-1},x_i) \rightarrow
(x_{i-1},q_i)$ and $(x_i,x_{i+1}) \rightarrow (q_i,x_{i+1})$ and
hence
$(x_0,x_1),(x_1,x_2),\dots,(x_{i-2},x_{i-1}),(x_{i-1},q_i),(q_i,x_{i+1}),(x_{i+1},x_{i+2}),\dots,(x_{n-1},x_n),$
\\ $(x_n,x_0)$ is also a circuit in $D \cup S$.

Thus $(q_0,x_1),(x_1,x_2),\dots,(x_{n-1},x_n),(x_n,q_0)$ is also a
circuit in $D \cup S$. Moreover if $q_0$ has a neighbor then this
neighbor is not adjacent to any of the $x_j$, $j \ne 0$. This give
rise to a disjoint connected components $H_0,H_1,\dots,H_n$ where
there is no edge between $H_i,H_j$, $i \ne j$. Moreover
$(x_i,x_j)$ and $(x_r,x_s)$, for $(i,j) \ne (r,s)$ are in
different components of $H^+$. These imply that $(x_0,x_n)$ is
also green and by adding $(x_0,x_n)$ into $D$ we don't form a
circuit as otherwise there would be an earlier circuit in $D$.

\noindent\textit{ Second we consider $\{12,23\} \in \F$.} Since
$x_nx_0$ is not an edge, and $(x_n,x_0)$ is in a non-trivial
component, $(x_n,x_0)$ is dominated by some pair $(x,y) \in S$
according to $\{13,23\}$. This means that there exists some $q_n$
such that $(x_n,q_n) \in S$ dominates $(x_n,x_0)$, $q_nx_0,q_nx_n
\in E(H)$. Note that $(x_n,q_n) \rightarrow (x_n,x_0) \rightarrow
(q_n,x_0)$. However $(x_n,q_n) \rightarrow (x_0,q_n)$ according to
$\{12,23\}$, a contradiction.

{\bf Case 2.} $\{12,23\} \in \F$, and $\{13,23\}, \{12,13\}
\not\in \F$.

\noindent\textit{First, suppose $x_nx_0$ is an edge of $H$.}

We show that $x_{n-1}x_n$ is an edge of $H$. Otherwise according
to the Algorithm and the assumption that none of the $\{13,23\},
\{12,13\}$ is in $\F$ then one of the patterns $\{12\}$, $\{23\}$,
$\{13\}$ is in $\F$. Now $x_{n-1}x_0$ is not an edge as otherwise
$(x_n,x_0) \rightarrow (x_{n-1},x_0)$ and hence we have a shorter
circuit.  However $(x_{n-1},x_n) \rightarrow (x_0,x_n)$ when
$\{23\} \in \F$ and $(x_n,x_0) \rightarrow (x_{n-1},x_0)$ when
$\{12\} \in \F$ and $(x_{n-1},x_n) \rightarrow (x_{n-1},x_0)$ when
$\{13\} \in \F$. In any case we obtain a shorter circuit. Thus
$x_{n-1}x_n$ is an edge and by applying the previous argument we
conclude that $x_ix_{i+1}$, $0 \le i \le n$ is an edge and hence
$x_0,x_1,\dots,x_n$ induce a clique.

Note that $(x_0,x_n)$ does not dominate any pair according to one
of the patterns $\{12\}$, $\{23\}$, $\{13\}$. For contradiction
suppose $(x_n,x_0)$ is dominated by some  pair $(x,y) \in S$
according to $\{23\}$ (the argument of other cases is analogous).
This means there exists $w$ such that $(w,x_0) \rightarrow
(x_n,x_0) \rightarrow (x_n,w)$. Now if $wx_j \not\in E(H)$ then
$(w,x_0) \rightarrow (x_j,x_0)$ and we obtain a shorter circuit.
Now $(x_{n-1},x_n) \rightarrow (x_{n-1},w) \rightarrow
(x_{n-1},x_0)$ implying a shorter circuit in $D \cup S$. Therefore
we may assume that $\{12,23\}$ is the only forbidden pattern in
$\F$.

Since $(x_n,x_0)$ dominates a pair in $S$ because of symmetry we
continue by assuming there exists $q_n$ such that $x_nq_n \in
E(H)$ and $x_0q_n \not\in E(H)$. In this case $(x_0,x_n)(q_n,x_n)$
is an edge of $H^+$ (symmetric). Now $(x_{n-1},x_n) \rightarrow
(q_n,x_n)$ and hence $(q_n,x_n) \in D \cup S$. On the other hand
we have $(x_n,q_n) \in S$, implying a circuit in $D \cup S$, a
contradiction.

\noindent\textit{Second, suppose $x_nx_0$ is not an edge of $H$.}

According to the algorithm we may assume one of the
$\{12\},\{23\},\{13\}$ is in $\F$. Now by the same argument as in
Claim \ref{independent}, $x_0,x_1,\dots,x_n$ induce an independent
set. By using symmetry and applying argument in Case B we conclude
$(x_0,x_n)$ can be added into $D$ without creating a circuit.

This completes the correctness proof of the algorithm and hence
the proof of Theorem \ref{main}. \qed \end{proof}

We now focus on Theorem \ref{main-bipartite} and the corresponding
algorithm.

\begin{proof} The argument that shows that the algorithm does not create a
circuit is very similar to the proof of Lemma
\ref{correctness-for-3}. However for the sake of completeness we
present the case when $\F =\{ \{12',21'\},\{12',21',22'\}\}$.  We
claim that the algorithm will never create a circuit, and hence
yield the desire ordering. We also prove an stronger version by
assuming that $H$ does not have an invertible pair. Then the
components of $H^+$ come in conjugate pairs $S, \overline{S}$. We
claim that the Algorithm for $\BORD_4$ does not creates a circuit.

Otherwise, suppose the addition of $S$ creates circuits for the
first time, and $(x_0,x_1), (x_1,x_2),$ \\ $ \dots, (x_n,x_0)$ is
a shortest circuit created at that time. Note that $n>1$,
according to the rules of the algorithm. Without loss of
generality, assume that the pair $(x_n,x_0)$ (and possibly other
pairs) lies in $S$.

We may assume that $(x_n,x_0)$ is in a non-trivial component $S$
of $H^+$, otherwise $(x_n,x_0)$ or $(x_0,x_n)$ is a sink, and
clearly no circuit is created during the preliminary stage when
sinks were handled. Thus $(x_n,x_0)$ dominates and is dominated by
some pair, which may be assumed to be the same pair, say
$(y_n,y_0)$.  We claim that each $x_i$ has a private neighbor
$y_i$ but not to any other $y_j$ with $j \neq i$.

Now $y_nx_j \not\in E(H)$ for $j \ne n, 0 $, as otherwise
$(y_n,y_0)$ dominates $(x_j,x_{0})$ and hence $(x_j,x_{0}) \in D
\cup S$ implying a shorter circuit in $D \cup S$. Moreover
$y_0x_j$, $j \ne n,0$ as otherwise $(y_n,y_0)$ dominates
$(x_{n},x_j)$ and hence $(x_{n},x_j) \in D \cup S$, again implying
a shorter circuit in $D \cup S$. For summary we have :

\begin{itemize}
\item $x_{0}y_0, x_ny_0 \in E(H)$ \item $x_jy_n \not\in E(H)$, $j
\ne  n$ \item $x_jy_0 \not\in E(H)$, $j \ne 0$.
\end{itemize}

If $(x_{0},x_{1})$ does not dominates any pair in $H^+$ then we
show that $(x_n,x_{1})$ does not dominates any pair in $H^+$
either, and hence by the rules of the algorithm $(x_n,x_{1})$ is
in $D$ already, contradicting the minimality of $n$. In contrary
suppose $(x_0,x_1)$ does not dominate any pair and $(x_n,x_{1})$
dominates some pair $(y_n,y_1)$ in $H^+$. Note that $x_1y_n$ is
not an edge. Now $x_{0}y_1$ is an edge and hence $(x_n,x_{0})$
dominate $(y_n,y_1) \in D \cup S$. Moreover $(y_n,y_1)$ dominates
$(x_n,x_{1})$ and hence $(x_n,x_1) \in D \cup S$ implying a
shorter circuit in $D \cup S$.

Therefore $(x_{0},x_{1})$ dominates $(y_0,y_1) \in D \cup S$.
Recall that none of the  $x_ny_0,x_0y_n,x_1y_0,x_1y_n$ is an edge
of $H$. Now $x_ny_1 \not\in E(H)$ as otherwise $(x_n,x_{0})$
dominates $(y_1,y_0)$ and hence $(y_1,y_0) \in D \cup S$. Now
$(y_1,y_0)$ dominates $(x_{1},x_0)$ implying that $(x_{1},x_0) \in
D \cup S$ and hence yielding a shorter circuit in $D \cup S$. We
conclude that $(x_0,x_{1})$ and $(x_n,x_{1})$ both are in non
trivial components of $H^+$, and $x_ny_n,x_{0}y_0,x_{1}y_1$ are
independent edges. By continuing this argument we conclude that
there are independent edges $x_0y_0,x_1y_1,...,x_ny_n$.

Let $X$ denotes the set of vertices of $H$ adjacent to all $x_i$,
or to all $y_i$. Note that $X$ is complete bipartite, as otherwise
$H$ contains a six-cycle (with diametrically opposite vertices in
$X$), and hence an invertible pair. Now we claim that any vertex
$v$ not in $X$ is adjacent to at most one $x_i$ (or $y_i$).
Otherwise we may assume $v$ is not adjacent to $x_{i-1}$ but is
adjacent to $x_i$ and $x_j$, which would imply that $(x_{i-1},
x_j)$ is in the same component as $(x_{i-1}, x_i)$, contradicting
the minimality of our circuit. (Thus it is impossible to have a
path of length two between $x_i$ and $x_j (i \neq j)$ without
going through $X$.) More generally, we can apply the same argument
to conclude that any path between $x_i$ and $x_j (i \neq j)$ must
contain a vertex of $X$. Therefore, $H - X$ has distinct
components $H_1, H_2, \dots H_n$, where $H_i$ contains $x_i$ and
$y_i$. We claim that the component of $H^+$ containing the pair
$(x_i,x_{i+1})$ consists of all pairs $(u,v)$ where $u$ is in
$H_i$ and $v$ is in $H_{i+1}$. This implies that $S$ does not
contain any $(x_i,x_{i+1})$ other that $(x_n,x_0)$; it also
implies that both $S$ and $\overline{S}$ are green.

Now consider the addition of $\overline{S}$. If it also leads to a
circuit $(z_0,z_1), \dots, (z_r,z_0)$, we may assume that
$z_0=x_n$ and $z_r=x_0$. By a similar argument we see that
$\overline{S}$ does not contain any other $(z_i,z_{i+1})$. This
means that $(x_0,x_1), \dots, (x_{n-1},x_n)=(x_{n-1},z_0),
(z_0,z_1), \dots (z_r,z_0)=(z_r, x_0)$ was an earlier circuit,
contradicting the assumption.

This completes the proof of Theorem \ref{main-bipartite} \qed
\end{proof}

\subsection{Obstruction Characterizations}

It is similarly the case that the constraint digraph offers a
unifying concept of an obstruction for graph classes $\BORD(\F),
\F \subseteq \B_4$. Namely, Theorem \ref{main-bipartite} characterizes
all these classes by the absence of a circuit in a strong
component of the constraint digraph. In some cases we can again
simplify the obstructions to a bipartite version of invertible
pairs.

An {\em invertible pair} of $H$ is a pair of vertices $u,v$ from
the same part of the bipartition such that both $(u,v)$ and
$(v,u)$ lie on the same directed cycle of $H^+$. Thus a circuit of
length two in a strong component of $H^+$ corresponds precisely to
an invertible pair.

For our first illustration we discuss the case of co-circular-arc
bigraphs. A {\em co-circular-arc bigraph} is a bipartite graph
whose complement is a circular arc graph. A complex
characterization of co-circular-arc graphs by seven infinite
families of forbidden induced subgraphs has been given in
\cite{trotter}, later simplified to a Lekerkerker-Boland-like
characterization by forbidden induced cycles and edge asteroids in
\cite{co-circular-pavol}. These graphs seem to be the bipartite
analogues of interval graphs, see \cite{co-circular-pavol}. One
reason may be that co-circular-arc bigraphs are precisely the
intersection graphs of $2$-directional rays \cite{japan}.

\vspace{2mm}

We observe the following simple characterization.

\begin{theorem}
Let $H=(B,W)$ be a bipartite graph. Then the following are
equivalent.
\begin{itemize}
\item [(1)] $H$ is a co-circular-arc bigraph.

\item [(2)] $H$ admits an $\F$-free ordering where $\F=
\{\{12',21'\},\{12',21',22'\}\}$.

\item[(3)] $H$ has no invertible pair.
\end{itemize}
\end{theorem}
\begin{proof} It was shown in \cite{arash-esa} that (1) and (2) are
equivalent. (In \cite{arash-esa} $\F$-free orderings are described
by an equivalent notion of  so-called min orderings.) According to
proof of Theorem \ref{main-bipartite} for $\F$ we assume that $H$
does not have an invertible pair. Therefore (2) and (3) are
equivalent and hence the theorem is proved. \qed \end{proof}

A bipartite graph $G=(V,U)$ is called {\em proper interval
bigraph} if the vertices in each part can be represented by an
inclusion-free family of intervals, and a vertex from $V$ is
adjacent to a vertex from $U$ if and only if their intervals
intersect. They are also known as bipartite permutation graphs
\cite{gregory,spinrad,spinrad1}.

\begin{theorem}
Let $H=(B,W)$ be a bipartite graph. Then the following are
equivalent.
\begin{itemize}
\item[(1)] $H$ is a proper interval bigraph

\item[(2)] $H$ admits an $\F$-free ordering where $\F=
\{\{12',21'\},\{12',21',22'\},\{11',12',21'\}\}$.

\item[(3)] $H$ does not have an invertible pair.
\end{itemize}
\end{theorem}
\begin{proof} It was noted in \cite{gregory} that $H$ admits an $\F$-free
ordering if and only if $H$ is a bipartite permutation graph
(proper interval bigraph). (In \cite{gregory} $\F$-free orderings
are described by an equivalent notion of min-max orderings.)
Therefore (1) and (2) are equivalent. It is easy to see that if no
strong component of $H^+$ contains a circuit of length two, i.e.
if $H$ has no invertible pair, then the Algorithm for $\BORD_4$
does not create a circuit. Therefore (2) and (3) are equivalent
and hence the theorem is proved. \qed \end{proof}

\section{Remarks and conclusions}

As noted earlier, Duffus, Ginn, and R\"{o}dl have found many examples
of NP-complete problems $\ORD(\F)$; in fact if $\F$ consists of a
single ordered pattern, they offered strong evidence that
$\ORD(\F)$ may be NP-complete as soon as the pattern is
2-connected. We offer just two simple examples to illustrate some
NP-complete cases.
\begin{proposition}
For every $k \ge 4$ there exists a set $\F \subseteq \F_k$ such that
$\ORD(\F)$ is NP-complete.
\end{proposition}
\begin{proof} We show that if $\F$ is a set of all forbidden patterns on $k$
vertices where each of them contains $\{12,23,34,\dots, (k-1)k\}$
as a subset, then $\ORD_k$ is NP-complete. We reduce the problem
to $(k-1)$-colorability. Let $H$ be an arbitrary graph. If $H$ is
$(k-1)$-colorable with color classes $X_1,X_2,\dots,X_{k-1}$, then
we put all the vertices in $X_i$ before all the vertices in
$X_{i+1}$, $1 \le i \le k-2$. This way we obtain an ordering of
the vertices and it is clear that it does not contain any of the
forbidden patterns in $\F$.

Now suppose there is an ordering $v_1,v_2,\dots,v_n$ of the
vertices in $H$ without seeing any forbidden pattern in $\F$. Let
$X_1$ be the set of vertices $v_j$, $1 \le j \le n$ that have no
neighbor before $v_j$. Now for every $2 \le i \le k-1$, let $X_i$
be the set of vertices $v_j$, $1 \le j \le n$ from set $V(H)
\setminus ( \cup_{\ell=1}^{\ell=i-1} X_{\ell})$ that have no
neighbor before $v_j$. Note that by definition each $X_i$, $1 \le
i \le k-1$ is an independent subset of $H$. Moreover $V(H)=
\cup_{\ell=1}^{\ell=i-1} X_{\ell}$ as otherwise we obtain $k$
vertices $u_1 < u_2 \dots < u_k$ where $u_ju_{j+1}$, $1 \le j \le
k-1$ is an edge of $H$ and hence we find a forbidden pattern from
$\F$. Thus $H$ is $(k-1)$-colorable.  \qed \end{proof}

We note that in Damaschke's paper \cite{dama} the complexity of
$\ORD(\F)$ was left open for $\F= \{12,23,34\}$. (However, other
folklore solutions for this particular case have been reported
since.)

In the case of bipartite graphs, we offer the following simple
example.

\begin{proposition}
$\BORD(\F)$ is NP-complete for set $\F=\{\{11',31',51'\} \}$.
\end{proposition}
\begin{proof} Let $M$ be a $m \times n$ matrix with entities $0$ and $1$.
Finding an ordering of the columns such that in each row there are
at most two sequences of consecutive $1$'s has been shown to be
NP-complete in \cite{golberg}. Now from an instance of a matrix
$M$ we construct a bipartite graph $H=(A,B,E)$ where $A$
represents the set of columns and $B$ represents the set of rows
in $M$. There is an edge between $a \in A$ and $b \in B$ if the
entry in $M$, corresponding to row $a$ and column $b$ is $1$. Now
if we were able to reorder to columns with the required property
we would be able to find the ordering of $H$ without seeing the
forbidden pattern in $F$ and vice versa. \qed \end{proof}

There are natural polynomial problems $\ORD(\F)$ for sets $\F$ of
larger patterns. For instance, strongly chordal graphs are
characterized as $\ORD(\F), \F \subseteq \F_4$, in \cite{martin}. In
fact, an algorithm similar to the one presented here can be
developed for this case. (We will be happy to communicate the
details to interested readers.)

There is a natural version of $\ORD(\F)$ for digraph patterns
$\F$. We are given an input digraph $H$ and a set $\F$ of
forbidden digraph patterns (each digraph pattern is an ordered
digraph). The decision problem asking whether an input digraph
admits an ordering without forbidden patterns in $\F$ is denoted
by $\DORD(\F)$. Let $\D_k$ denote the collection of sets $\F$ of
digraph patterns with $k$ vertices. The problem $\DORD_k$ asks,
for an input $\F \subseteq \D_k$ and a digraph $H$, whether or not $H$
has an $\F$-free ordering.

The algorithms in \cite{adjust-interval,arash} illustrate two
cases where $\DORD(\F)$ problems have been solved by algorithms
similar to the algorithm for $\ORD_3$, and the obstructions
characterized as invertible pairs. (The problems in
\cite{adjust-interval,arash} are not presented as $\ORD(\F)$, but
they can easily be so reformulated.) We believe many other digraph
problems can be similarly handled. In fact we wonder whether the
problem $\DORD(\F)$ is polynomial for every set $\D \in \D_3$ .

We conjecture that for every set $\F$ of forbidden patterns,
$\ORD(\F)$ is either polynomial or NP-complete.

\end{document}